\documentclass{ifacconf}

\usepackage{graphicx} 

\usepackage{caption}
\usepackage{subcaption}
\usepackage{color}
\usepackage{amsmath} 
\usepackage{amssymb}
\usepackage{overpic}
\usepackage{wrapfig}
\usepackage{natbib}

\usepackage{xcolor}
\usepackage{soul} 

\makeatletter
\newenvironment{proof}[1][Proof]{%
  \par\noindent\textit{#1. }\rmfamily}{\hfill$\square$\par}
\makeatother

\newtheorem{proposition}{Proposition}
\newtheorem{remark}{Remark}

\begin{document}
\begin{frontmatter}

\title{A Geometric Task-Space Port-Hamiltonian Formulation for Redundant Manipulators} 


\author[RaM]{Federico Califano} 
\author[Sapienzastud]{Camilla Rota} 
\author[RaM]{Riccardo Zanella} 
\author[RaM,Sapienza]{Antonio Franchi}

\address[RaM]{Robotics and Mechatronics (RaM), University of Twente, 
Drienerlolaan 5, 7522 NB Enschede, The Netherlands, 
(e-mail: schol@r-zanella.eu, f.califano@utwente.nl, schol@r-franchi.eu)}

\address[Sapienzastud]{Student at Sapienza University of Rome}

\address[Sapienza]{DIAG, Sapienza University of Rome,
00185 Rome, Italy}

\begin{abstract}                
We present a novel geometric port-Hamiltonian formulation of redundant manipulators performing a differential kinematic task $\eta=J(q)\dot{q}$, where $q$ is a point on the configuration manifold, $\eta$ is a velocity-like task space variable, and $J(q)$ is a linear map representing the task. The proposed model emerges from a change of coordinates from canonical Hamiltonian dynamics, and decomposes the standard Hamiltonian momentum variable into a task-space and a null-space component. Properties of this model and relation to Lagrangian formulations present in the literature are highlighted. Finally, we apply the proposed model in an \textit{Interconnection and Damping Assignment Passivity-Based Control} (IDA-PBC) design to stabilize and shape the impedance of a 7-DOF Emika Panda robot in simulation.
\end{abstract}


\end{frontmatter}

\section{Introduction}

Port-Hamiltonian (pH) systems are nonlinear affine systems encoding a rich set of system-theoretical and behavioral properties, conveniently modeling networks of physical systems in any physical domain \citep{Duindam2009ModelingSystems,Rashad2020TwentyReview}. PH describe energy exchange through \textit{power ports}: pairs of dual variables on which energy flows. These ports provide the mechanism for interconnection with external systems. From a control-theoretic perspective, a pH system encodes desirable properties such as passivity, stability, and general energetic behaviors that can serve as performance metrics in complex tasks \citep{vanderSchaftL2,ortega2002interconnection,Ortega2001PuttingControl}.

We present a novel geometric pH model for fully actuated redundant manipulators \citep{Park_Chang_Lee_2001,Dietrich2015} in the task space. In particular, distinct power ports defined in the task-space and in the null-space of the task, emerge. Redundancy is meant with respect to a differential kinematic task $\eta=J(q)\dot{q}$, where $q$ is a configuration of the manipulator, and $J(q)$ linearly maps the joint velocity $\dot{q}$ to a lower dimensional, velocity task-space variable $\eta$. The geometric aspect is encoded in a rigorous treatment of both the differential geometric and linear algebraic aspects of the theory. Other redundancy studies present in the literature are treated in a Lagrangian framework \citep{Michel2022,Michel2020,Ott2015,Dietrich2015,Ott2011} and it is interesting to comment and discuss the differences and advantages emerging in the proposed Hamiltonian context. In the Hamiltonian framework, recently in \citep{Scherpen2025} a task-space impedance shaping procedure for manipulators has been proposed in the pH framework performing a change of coordinates on the Hamiltonian coordinates to get a target system in task space momentum coordinates. However, that work does not address redundancy, which is the key technical challenge tackled in this paper.

The proposed model enables convenient representation of a redundant manipulator interacting with an external system and allows analysis of energy flows through the interconnection induced by the physical interaction, which structurally happens in the task-space. Furthermore the proposed model can be conveniently used in control strategies developed in the pH framework \citep{vanderSchaftL2,Duindam2009ModelingSystems}, such as \textit{Interconnection and Damping Assignment Passivity-Based Control} (IDA-PBC) \citep{ortega2004interconnection,ortega2002interconnection}, \textit{de facto} opening to the possibility of exploring new task-space control possibilities for redundant manipulators. In this respect, a further contribution of this work is the use of the presented model along an IDA-PBC design procedure, which is proposed to stabilize and shape the impedance of a 7-DOF Franka Emika Panda robot.


\section{Setting}
\label{sec:setting}

\subsubsection*{Notation.} Given a vector space $W$, we indicate with $W^*$ its dual space. We represent both covectors $\beta\in W^{*}$ and vectors $w\in W$ in columns of their components, so that $\beta^{\top}w$ returns the application of $\beta$ on $w$. In particular, given the $n$-dimensional configuration space of a mechanical system $Q$, we indicate with $T_qQ$ and $T_q^*Q$ respectively its tangent and co-tangent spaces at the point $q\in Q$, so that the scalar mechanical power $\tau^{\top}\dot{q}$ results from the application of the generalized force covector $\tau\in T_q^*Q$ to the velocity vector $\dot{q}\in T_q Q$. Given a linear map $A:W \to V$, we indicate with $\textrm{Ker}(A)$ its kernel, and with $\textrm{Im}(A)$ its image. Given a vector space $W$ and a vector subspace $S\subseteq W$, we indicate with $\textrm{Ann}(S)=\{\beta \in W^* | \beta^{\top}s=0, \forall s \in S\}$ the annihilator of $S$. If $W$ is equipped with an inner product (or metric) $M:W \times W \to \mathbb{R}$ (calculated as $M(w_1,w_2)=w_1^{\top}Mw_2$ for $w_1,w_2 \in W$), we indicate with $S^{\perp_{M}}=\{w \in W | w^{\top}Ms=0, \forall s \in S\}$ the $M$-orthogonal complement of $S$. 
Given a smooth function $f(x)$ we indicate with $\partial_xf$ the gradient covector, represented as a column vector, and with $\partial_x^2f$ its symmetric Hessian matrix. $I_n$ is the $n \times n$ identity matrix.

We consider a \textit{differential task} defined as a linear map $J(q):T_qQ \to \mathcal{X}$ as
\begin{equation}
\label{eq:task}
    \eta=J(q)\dot{q}
\end{equation}
with $\mathcal{X}$ a $m$-dimensional vector space with $m\leq n$ and $\textrm{rank}(J(q))=m$.
This task can represent, for example, analytic differential kinematics, in which case $J(q)$ is the differential of a forward kinematic map $x=f(q)$ and $\eta=\dot{x}$, but can also encode \textit{non integrable} tasks, e.g., the case in which $J(q)$ is the geometric Jacobian, which maps $\dot{q}$ to the geometric \textit{twist} containing the angular velocity as rotational velocity component \citep{lynch2017modern}.
We consider a fully actuated mechanical system:
\begin{equation}
\label{eq:robot}
M(q)\ddot{q}+h(q,\dot{q})+\partial_qV(q)= \tau
\end{equation}
with $(q \,\,\, \dot{q})\in Q \times T_qQ$ the standard Lagrangian coordinates on the tangent bundle of the $n$-dimensional configuration manifold $Q$, $M(q)=M^{\top}(q)\in \mathbb{R}^{n\times n}$ represents the inertia tensor (an inner product field) on $Q$, $V(q)$ smoothly maps the configuration to the conservative potentials in the system (e.g., elastic and gravitational), $\tau \in T_q^*Q$ are the generalized actuation forces, and $h(q,\dot{q})$ collects the nonlinear terms present in the Euler-Lagrange formulation (\ref{eq:robot}).
We will also use the equivalent port-Hamiltonian formulation:
\begin{align}
\label{eq:robotph}
    \begin{pmatrix}
        \dot{q}\\ \dot{p}
    \end{pmatrix}& = \begin{pmatrix}0&I_n\\ -I_n&0\end{pmatrix} \begin{pmatrix} \partial_qH\\ \partial_pH\end{pmatrix} + \begin{pmatrix} 0 \\ I \end{pmatrix} \tau\\
    y &= \begin{pmatrix}
        0 & I
    \end{pmatrix}
    \begin{pmatrix}
        \partial_q H \\ \partial _p H
    \end{pmatrix}=\dot{q} 
\end{align}
with \textit{Hamiltonian}:
\begin{equation}
\label{eq:Hamiltonian}
  H(q,p)=K(q,p)+ V(q)
\end{equation} 
representing the mechanical energy given by the sum of kinetic energy 
\begin{equation}
\label{eq:kinenergy}
    K(q,p)=\frac{1}{2}p^{\top}M^{-1}(q)p
\end{equation}
and potential energy $V(q)$.
Here the standard Hamiltonian coordinates are $(q \,\,\, p)\in Q \times T_q^*Q$ and the two formulations are connected by 
\begin{equation}
\label{eq:lagrham}
    p=M(q)\dot{q}.
\end{equation}
From the Hamiltonian formulation it is easy to see that the mechanical power balance representing energy conservation in (\ref{eq:robotph}) takes the form:
\begin{equation}
\label{eq:passivity}
  \dot{H}=\tau^{\top}\dot{q}.  
\end{equation}
In system-theoretic terms, (\ref{eq:passivity})  and embodies the property of \textit{passivity} \citep{vanderSchaftL2} of (\ref{eq:robotph}) with storage function $H$ and power conjugated variables (or \textit{port variables}) $(\tau ,\dot{q})\in T_q^*Q\times T_qQ$.
Dissipative effects can be added without any technical problem nor conceptual insight to all the results presented in this paper, and as such will be omitted to enhance readability.

\section{Energy and power decomposition in redundant manipulators}
\label{sec:geometry}

Using (\ref{eq:task}) and the metric $M(q)$ we construct the following decomposition of velocities $\dot{q}\in T_qQ$ and generalized force variables $\tau \in T_q^{*}Q$. Although this section largely reviews known results and uses standard mathematical tools, the proposed duality-based decomposition on the space of generalized forces is, to the best of the authors’ knowledge, novel in its formulation in a robotic context.

\subsubsection{Decomposition of velocities.}

The metric $M(q)$ defines the kinetic energy:
\begin{equation}
\label{eq:kinEnergy}
    K(q,\dot{q})=\frac{1}{2}\dot{q}^{\top}M(q)\dot{q}
\end{equation}
where we slightly abused notation using the same symbol as in the Hamiltonian version (\ref{eq:kinenergy}).
As a matter of fact, $\forall q\in Q$ the metric $M(q)$ induces the geometric decomposition 
\begin{equation}
\label{eq: velocityDecomposition}
T_q Q = \textrm{Ker}(J(q)) \oplus \textrm{Ker}^{\perp_{M}}(J(q)),
\end{equation}
where $\textrm{Ker}^{\perp_{M}}(J(q))$ is the orthogonal complement of $\textrm{Ker}(J(q))$ with respect to $M(q)$.
In other words, $\forall \dot{q} \in T_qQ$ there is a unique decomposition
\begin{equation}
\label{eq:decomposition}
    \dot{q}=\nu+v
\end{equation}
where $\nu \in \textrm{Ker}(J(q))$ and $v \in \textrm{Ker}^{\perp_{M}}(J(q))$, where it holds 
\[
\eta=J(q)\dot{q}=J(q)(v+\nu)=J(q)v.
\]
The kinetic energy (\ref{eq:kinEnergy}), using (\ref{eq:decomposition}) and the ortogonality condition in $\textrm{Ker}^{\perp_{M}}(J(q))$, results in:
\begin{equation}
\label{eq:kineticenergysplit}
    K=\frac{1}{2}v^{\top} M(q) v+\frac{1}{2}\nu^{\top} M(q) \nu,
\end{equation}
i.e., the mixed terms vanish.
The second term in (\ref{eq:kineticenergysplit}) is the kinetic energy in the null space of the task:
$
K_{\nu}(q,\dot{q}):=\frac{1}{2}\nu^{\top} M(q) \nu.
$
To understand the remainder term $K_t(q,\dot{q}):=\frac{1}{2}v^{\top} M(q) v$
as kinetic energy in the task space, we have to use the \textit{dynamically consistent} pseudoinverse \citep{Dietrich2015}:
\begin{equation}
    \label{eq:Jm}
    J_M^{\#}=M^{-1}J^{\top}(JM^{-1}J^{\top})^{-1}
\end{equation}
which maps a task velocity $\eta \in \mathcal{X}$ to $\textrm{Ker}^{\perp_{M}}(J(q))\ni v$ (see Fig. \ref{fig:placeholder}).
Indeed, using $v=J^{\#}_M\eta$, one obtains:
\[K_t(q,\dot{q})=\frac{1}{2} \eta^{\top} (J_M^{\#})^{\top}M(J_M^{\#}) \eta.
\]
Using (\ref{eq:Jm}) the latter becomes
\begin{equation}
\label{eq:kinTask}
K_t(q,\dot{q})=\frac{1}{2}\eta^{\top} \Lambda(q) \eta
\end{equation} where 
\begin{equation}
\label{eq:mobept}
    \Lambda:=(JM^{-1}J^{\top})^{-1}.
\end{equation}

The physical interpretation of $\Lambda(q)$ emerges when describing the robot dynamics in the task space and accounting for external interactions $F_{ext}\in \mathcal{X}^*$ entering at a torque level in (\ref{eq:robot}), i.e., the right hand side of (\ref{eq:robot}) becomes $\tau +J(q)^{\top}F_{ext}$.
Differentiating (\ref{eq:task}) to $
\dot{\eta}=\dot{J}\dot{q}+J\ddot{q}$
and then substituting (\ref{eq:robot}) with this variation, one obtains:
\begin{equation}
\label{eq:robotTask}
\dot{\eta}=\Lambda^{-1} F_{ext}+JM^{-1}(\tau -h-\partial_qV)+\dot{J}\dot{q}
\end{equation}
 In static conditions $\dot{q}=0$ and $\tau=\partial_q V$, the latter results in $\Lambda(q) \dot{\eta}=F_{ext}$
i.e., $\Lambda(q)$ is interpreted as the effective inertia in the end effector space. When $J$ is the analytic or geometric manipulator Jacobian, $\Lambda$ is referred to as \textit{mobility end-point tensor} \citep{Khatib}, but the construction is valid for any differential task (\ref{eq:task}). 

\begin{remark}
Using other pseudo-inverse relations to invert (\ref{eq:task}), such as the ``regular" Moore-Penrose pseudo-inverse $J^{\#}=J^{\top}(JJ^{\top})^{-1}$), would produce a vector $J^{\#}\eta$ that does not belong to the space $\textrm{Ker}^{\perp_{M}}(J)$. As a consequence we would loose the physical null-space/task-space kinetic energy decomposition (\ref{eq:kineticenergysplit}), crucial in the following construction. Regular pseudo-inverse relations of (\ref{eq:task}), as often used in inverse differential kinematics algorithms, generate velocity decompositions that do not carry clear physical meanings.
\end{remark}
\subsubsection{Decomposition of generalized forces.}
From the differential task (\ref{eq:task}) we canonically have the dual map $J^{\top}: \mathcal{X}^*\to T_q^*Q$ as $\tau = J^{\top} \sigma$ where covectors $\sigma \in \mathcal{X}^*$ represent generalized force conjugated with the velocity-like vectors $\eta$.
We introduce now a decomposition on the generalized force space $T_q^*Q$ which carries an important physical meaning and is, in a linear algebraic sense, \textit{dual} to the velocity decomposition (\ref{eq: velocityDecomposition}):
\begin{equation}
\label{eq:ttd}
T_q ^* Q = \textrm{Ann}(\textrm{Ker}^{\perp_{M}}(J(q))) \oplus \textrm{Ann}(\textrm{Ker}(J(q)))
\end{equation}
i.e., we chose to decompose the torque space with a direct sum of the annihilators of the spaces which form the canonical velocity decomposition.
The fundamental subspaces theorem of linear algebra states:
\begin{equation}
\label{eq:ftla}
    \textrm{Ann}(\textrm{Ker}(J))=\textrm{Im}(J^{\top}).
\end{equation}

Furthermore, it is easy to show that: 
\begin{equation}
    \label{eq:spaceId}
      \textrm{Ann}(\textrm{Ker} ^{\perp_M}(J))=\textrm{Ker}(J M^{-1}),
\end{equation}
which fixes the decomposition (\ref{eq:ttd}) to
\begin{equation}
\label{torqueDecomposition}
T_q ^* Q = \textrm{Ker}(J M^{-1})   \oplus \textrm{Im}(J^{\top}).
\end{equation}

As a consequence, $\forall \tau \in T_q^* Q$ can be expressed uniquely as
\begin{equation}
\label{eq:torquesplit}
    \tau=\tau_0+\tau_F
\end{equation}
where $\tau_0 \in \textrm{Ker}(J M^{-1})$ and $\tau_F \in  \textrm{Im}(J^{\top})$. The physical meaning of this decomposition is clear once (\ref{eq:torquesplit}) is used in (\ref{eq:robotTask}): the $\tau_0$ component does not affect the task space dynamic equation as it is pre-multiplied by $JM^{-1}$.

\subsection{Decomposition of power}

Using the presented decompositions on velocity and force space the power balance (\ref{eq:robotph}) can be written as the sum of two terms:
\begin{equation}
\label{eq:powerDecomposition}
\tau^{\top}\dot{q}=(\tau_0+\tau_F)^{\top}(v+\nu)=\tau_0^{\top}\nu+\tau_F^{\top}v
\end{equation}
as the other terms are clearly zero by the definition of annihilators in the force decomposition (\ref{eq:ttd}). In other words, similarly of what happened for kinetic energy, we can geometrically distinguish a power components at the task level $\tau_F^{\top} v$, and one in the null space of the task $\tau_0^{\top} \nu$.
\begin{remark}
It is worth noticing that $\tau_F^{\top}\nu=0$ is true independently of the metric $M$, as (\ref{eq:ftla}) does not contain metric information. 
\end{remark}

The discussed decompositions can be computed using so called \textit{dynamic consistent} projectors (see Fig. \ref{fig:placeholder} and \cite{Dietrich2015}).
In particular, the velocity projector: $P=J^{\#}_MJ$
(resp. $(I-P)$), projects any $\dot{q}$ in the space $\textrm{Ker}^{\perp_{M}}(J)\ni v$ (resp. in the space $\textrm{Ker}(J)\ni \nu$), so that $P \dot{q}$ (resp. $(I-P)\dot{q}$) is a coordinate based representations in $n$ components, i.e., using the same basis of $T_qQ$ where $\dot{q}$ is expressed, of the vector $v$ (resp. $\nu$).
Dually, for the torques, we have the projector: $P^{T}=J^{\top}(J^\#_M)^{\top}$
(resp. $(I-P^{T})$), projects any $\tau$ in the space $\textrm{Im}(J^{\top})\ni \tau_F$ (resp. in the space $\textrm{Ker}(JM^{-1})\ni \tau_0$), so that $P^{T} \tau$ (resp. $(I-P^{T})\tau$) is a coordinate based representations in $n$ components, i.e., using the same basis of $T^*_qQ$ where $\tau$ is expressed, of the covector $\tau_F$ (resp. $\tau_0$).


\section{A port-Hamiltonian formulation of redundant manipulators in task space}
\label{sec:ph}

We now present a novel port-Hamiltonian formulation for manipulators along differential kinematic tasks (\ref{eq:task}) which displays the presented geometric decomposition of energy and power. The strategy is to perform a change of coordinates on the standard Hamiltonian state variables in (\ref{eq:robotph}) that let the structure emerge. The change of coordinates is inspired by a construction presented in \citep{park2} at a Lagrangian level, which we first report and then reinterpret and utilize in the Hamiltonian setting.
First we define the \textit{extended task velocity}:
\begin{equation}
\label{eq:taskVelocity}
\eta_e = \bar{J}(q)\,\dot{q},
  \qquad
  \bar{J}(q)=\begin{bmatrix}J(q)\\ N(q)\end{bmatrix}\in\mathbb{R}^{n\times n},
\end{equation}
and the \textit{extended mobility end point tensor}:
\begin{equation}
\label{eq:barGamma}
  \bar{\Lambda}(q) = (\bar{J}(q) M^{-1}(q) \bar{J}^{\top}(q))^{-1}\in \mathbb{R}^{n\times n}.
\end{equation}

The matrix $N(q) \in \mathbb{R}^{(n-m)\times n}$ is chosen such that $\bar{J}$ is invertible (and as a consequence $\bar{\Lambda}$ always exists) and such that $\bar{\Lambda}$ becomes block-diagonal, i.e., $\bar{\Lambda}=\textrm{blockdiag}(\Lambda_t,\Lambda_{\nu})$ where $\Lambda_t\in \mathbb{R}^{m \times m}$ and $\Lambda_{\nu}\in \mathbb{R}^{(n-m) \times (n-m)}$.
With this construction the extended task velocity gets decomposed into the task velocity $\eta=J\dot{q}$ as in (\ref{eq:task}) and the remaining $(n-m)$ components forming a vector $\nu=N\dot{q}$ representing the \textit{null space velocities} to the task encoded by $J$. In practice the matrix $N$ is constructed as follows: one chooses a full row rank null space base matrix of $J$, i.e., a matrix $Z(q)$ such that $JZ^{\top}=0$ and then one defines $N(q)=(Z(q)M(q)Z(q)^\top)^{-1}Z(q)M(q)$, which guarantees both the required properties for $N$. From a geometric perspective (see Fig. \ref{fig:placeholder} and the discussion about projectors in Sec. \ref{sec:geometry}) $N$ and $Z$ represent dual maps similar to the dynamically consistent projectors $(I-P)$ and $(I-P^{\top})$, where the difference is that the formers represent the result on a $(n-m)$-dimensional basis of respectively $\textrm{Ker}(J)$ and $\textrm{Ker}(JM^{-1})$, while the latters on a $n$-dimensional basis of $T_qQ$ and $T_q^*Q$\footnote{This fact justifies the slight abuse of notation for the same symbol $\nu \in \textrm{Ker}(J)$ used in both representations.}.
As noted in \citep{Ott2008}, the extended velocity $\eta_e$ is in general \emph{non–integrable}. For the present setting this observation is crucial because it prevents the formulation of a conventional Hamiltonian system directly in task space coordinates with position variable $\int_t \eta_e$\textrm{dt}. To preserve a rigorous port-Hamiltonian structure while retaining task-level information, the key idea is to adopt an hybrid joint/task space coordinates in which configuration coordinates are kept in joint space $q$ and new momentum variables are introduced at extended task-space level. 
Specifically, we introduce the following change of coordinates $\Phi$ from the canonical Hamiltonian coordinates in (\ref{eq:robotph}) to a new state variable $z$ as:
\begin{equation}\label{eq:hybrid-change-intro}
  z =\Phi \begin{pmatrix}
        q \\
       p
  \end{pmatrix}
  := \begin{pmatrix}
      q \\ 
      \bar{J}^{-\top}(q)\,p
  \end{pmatrix},
\end{equation}
i.e., the new variable $z$ copies the $q$ component present in standard Hamiltonian formulation and generates an \textit{extended momentum variable}
$\pi_e:=\bar{J}^{-T}(q)p.$ 
From this definition and using (\ref{eq:taskVelocity}), (\ref{eq:barGamma}) and (\ref{eq:lagrham}), it is straightforward to prove and implement the decomposition into a \textit{task-space momentum} and \textit{null-space momentum}:
\begin{equation}
\label{eq:momentumsplit}
   \pi_e=\bar{\Lambda}(q)\eta_e=\begin{pmatrix}
        \Lambda_t\eta \\
       \Lambda_{\nu}\nu
  \end{pmatrix} =: \begin{pmatrix}
        \pi \\
       \pi_{\nu}
  \end{pmatrix}.
\end{equation}
The physical meaning of this variables follows by expressing the kinetic energy (\ref{eq:kinenergy}) in the new coordinates:
\begin{align*}
\bar{K}(q,\pi_e):&=K(\Phi^{-1}(q,\pi_e))=\frac{1}{2}\pi_e^{\top}\bar{\Lambda}^{-1}(q)\pi_e.
\end{align*}
Using the fact that $\bar{\Lambda}(q)$ (and as consequence its inverse) is blockdiagonal, the momentum decomposition in (\ref{eq:momentumsplit}) acquires a physical interpretation in terms of kinetic energy splitting as:
\[
\bar{K}(q,\pi_e)=\underbrace{\frac{1}{2}\pi^{\top}\Lambda_t^{-1}(q)\pi}_{\bar{K}_t(q,\pi)}+\underbrace{\frac{1}{2}\pi_\nu^{\top}\Lambda_\nu^{-1}(q)\pi_\nu}_{\bar{K}_\nu(q,\pi_\nu)}.
\]
This is exactly, at a momentum level, the kinetic energy decomposition described in (\ref{eq:kineticenergysplit}): the task space kinetic energy (\ref{eq:kinTask}) is here expressed as a quadratic form on the task space momentum variable $\pi$, while the null space kinetic energy on the null space momentum variable $\pi_\nu$.
We are now ready to present the full port-Hamiltonian system in the new coordinates $q$, $\pi$ and $\pi_\nu$.
\\
\begin{proposition}
Given a mechanical system (\ref{eq:robotph}) and differential task (\ref{eq:task}), the port-Hamiltonian system in the variables defined by (\ref{eq:hybrid-change-intro}) and (\ref{eq:taskVelocity}) is given by:
\begin{align}\label{extpH}
   \begin{pmatrix}
       \dot q \\ \dot \pi \\ \dot \pi_\nu
   \end{pmatrix} &= \underbrace{\begin{pmatrix}
        0 & J^\#_M & Z^\top\\
        -J^{\#^\top}_M & G_{11} & G_{12}\\
        -Z & -G_{12}^{\top} & G_{22}
        \end{pmatrix}}_{J_z} \begin{pmatrix}
        \partial_q \bar{H} \\
        \partial _{\pi} \bar{H}\\
        \partial _{\pi_\nu} \bar{H}
    \end{pmatrix} + \begin{pmatrix}
        0 \\ \sigma \\ 0
    \end{pmatrix}+ \begin{pmatrix}
        0 \\ 0 \\ \tau_0
    \end{pmatrix}\\ 
    y&=\begin{pmatrix}
        0 & I_m & 0
    \end{pmatrix} \begin{pmatrix}
        \partial _q \bar{H}\\  \partial _\pi \bar{H} \\ \partial _{\pi_\nu} \bar{H}
    \end{pmatrix}=\partial _\pi \bar{H}=\eta \\
    y_\nu &=\begin{pmatrix}
        0 & 0 & I_{n-m}
    \end{pmatrix} \begin{pmatrix}
        \partial _q \bar{H}\\  \partial _\pi \bar{H} \\ \partial _{\pi_\nu} \bar{H}
    \end{pmatrix}=\partial _{\pi_\nu} \bar{H}=\nu 
\end{align}
with Hamiltonian
    \[\bar{H}(z):=H(\Phi^{-1}(z))=\frac{1}{2} \pi^\top \Lambda_t^{-1}(q) \pi + \frac{1}{2} \pi_\nu \Lambda^{-1}_\nu(q) \pi_\nu + V(q)\] and
   $ G_{11}=-G_{11}^{\top}=A_t  J^{\#}_M- J^{\#^\top}_MA_t^\top$, $
    G_{12}=A_t  Z^\top- J^{\#^\top}_MA_n^\top$, $
    G_{22}=-G_{22}^{\top}=A_n  Z^\top- ZA_n^\top$,
with $A_t=(\partial_q J_M^{\#^\top})\bar{J}^\top \pi_e+J_M^{\#^\top}(\partial _q \bar{J}^\top)\pi_e$ and $A_n=(\partial_q Z)\bar{J}^\top \pi_e+Z(\partial _q \bar{J}^\top)\pi_e$.
Furthermore $\begin{bmatrix}
        J^{\#}_M & Z^\top
    \end{bmatrix}=\bar{J}^{-1}$, $\sigma =(J^{\#}_M)^{\top}\tau$ and $\tau_0=Z\tau$. The null space port variables $\tau_0 \in \textrm{Ker}(JM^{-1})$ and $\nu\in \textrm{Ker}(J)$ are exactly the velocity components and generalized force components in the decompositions (\ref{eq: velocityDecomposition}) and (\ref{torqueDecomposition}), while the task space power port variables $\sigma\in \mathcal{X}^{*}$ and $\eta\in \mathcal{X}$ are isomorphic to $\tau_F\in \textrm{Im}(J^{\top})$ and $v\in \textrm{Ker}^{\perp_{M}}(J)$ in (\ref{eq:powerDecomposition}) and related by $\tau_F^{\top}v=\sigma^{\top}\eta$, i.e., their pairing is the mechanical power at the task port.
\end{proposition}

\begin{proof}
The expression $\begin{bmatrix}
        J^{\#}_M & Z^\top
    \end{bmatrix}=\bar{J}^{-1}$ follows from the standard construction of $\bar{J}$ (see \cite{Ott2015}).
Indicating with $h=
        (q^{\top} 
       p^{\top})^{\top}$ the standard Hamiltonian coordinates in (\ref{eq:robotph}), the port-Hamiltonian structure is preserved under the change of coordinates (\ref{eq:hybrid-change-intro}). In particular we have $\dot{z}=F \dot{h}$ with 
  \begin{equation*}
    F=\frac{\partial \Phi}{\partial h}=
    \begin{pmatrix}
        \frac{\partial q}{\partial q} & \frac{\partial q}{\partial p} \\
        \frac{\partial \pi_e}{\partial q} & \frac{\partial \pi_e}{\partial p}
    \end{pmatrix} = \begin{pmatrix}
        I_n & 0 \\
        \frac{\partial \bar{J}^{-\top}}{\partial q} p & \bar{J}^{-\top}
    \end{pmatrix}.
\end{equation*}
Substituting (\ref{eq:robotph}), and using the chain rule $\partial_h H=F^{\top}\partial_z \bar{H} $ one obtains: \[
\dot{z}=F \dot{h} =\underbrace{\Bigg(F\begin{pmatrix}0&I_n\\ -I_n&0\end{pmatrix}F^\top\Bigg)}_{J_z} \partial_{z}\bar{H}+  \begin{pmatrix}
   0 \\ \bar{J}^{-T}
\end{pmatrix}  \tau.\] The structure (\ref{extpH}) follows by calculation, taking care of expressing $F$ as function of the new coordinates $z$.
\end{proof}

We list properties of interest of the presented model:
\textbf{(i)} The model can be seen as a port-Hamiltonian counterpart of the Lagrangian model described in \cite{Ott2008}. In particular, when $J$ represents the Jacobian matrix of a manipulator, the evolution equation on the task-space momentum $\pi$ is an Hamiltonian counterpart of the classical task space dynamics ``à la Khatib" (\ref{eq:robotTask}) \citep{Khatib};
\textbf{(ii)} The structural properties of the pH system makes energy-based analysis much easier than in Lagrangian descriptions. As an example, in case of autonomous system $\tau_0=0$, $\sigma=0$, the energy conservation $\dot{\bar{H}}=(\partial_z\bar{H})^{\top}\dot{z}=(\partial_z\bar{H})^{\top}J_z\partial_z\bar{H}=0$ follows immediately by skew-symmetry of $J_z$; \textbf{(iii)} The power balance encoded in the pH model (\ref{extpH}) \begin{equation}
\label{eq:passivitysplit}
\dot{\bar{H}}=\sigma^{\top}\eta+\tau_0^{\top}\nu
        \end{equation}
        displays more information than the one in the standard coordinates (\ref{eq:passivity}).
        In particular the right hand side in (\ref{eq:passivitysplit}) mimics exactly the power decomposition (\ref{eq:powerDecomposition}), explicitly displaying a task port and a null space port in which the variables carry the meaning described in Sec. \ref{sec:geometry}. From a system theoretic perspective, this formulation encodes passivity of the system with respect to the storage function $\bar H$ and two distinct pairs of power conjugated variables $(\sigma,\eta)\in \mathcal{X}^*\times \mathcal{X}$ (the task port variables, which can be used for modelling external interaction at the task port) and $(\tau_0,\nu)\in \textrm{Ker}(JM^{-1})\times \textrm{Ker}(J)$ (the null space port variables);
    \textbf{(iv)} The off-diagonal terms $G_{12}$ and $-G_{12}^\top$ introduce internal couplings between the momentum components, explicitly revealing how kinetic energy is exchanged between task and null-space dynamics;
    \textbf{(v)} As highlighted in Fig. \ref{fig:placeholder}, the change of coordinates $\Phi$ in (\ref{eq:hybrid-change-intro}) carries a clear geometric meaning: it lifts the momentum space $T_q^*Q\ni p$ into the new momentum space $\mathcal{X}^*\times \textrm{Ker}(JM^{-1})\ni (\pi,\pi_\nu)$ where the first slot contains information at a task level while the second at a null space level.

To summarize, this formulation highlights the passivity of the system with respect to a decomposed task/null space power port and displays the energetic coupling between joint coordinates, task variables, and null-space dynamics. Furthermore, it opens to energy-based designs developed in the port-Hamiltonian framework.
\begin{figure}[t]
    \centering
\includegraphics[width=\columnwidth]{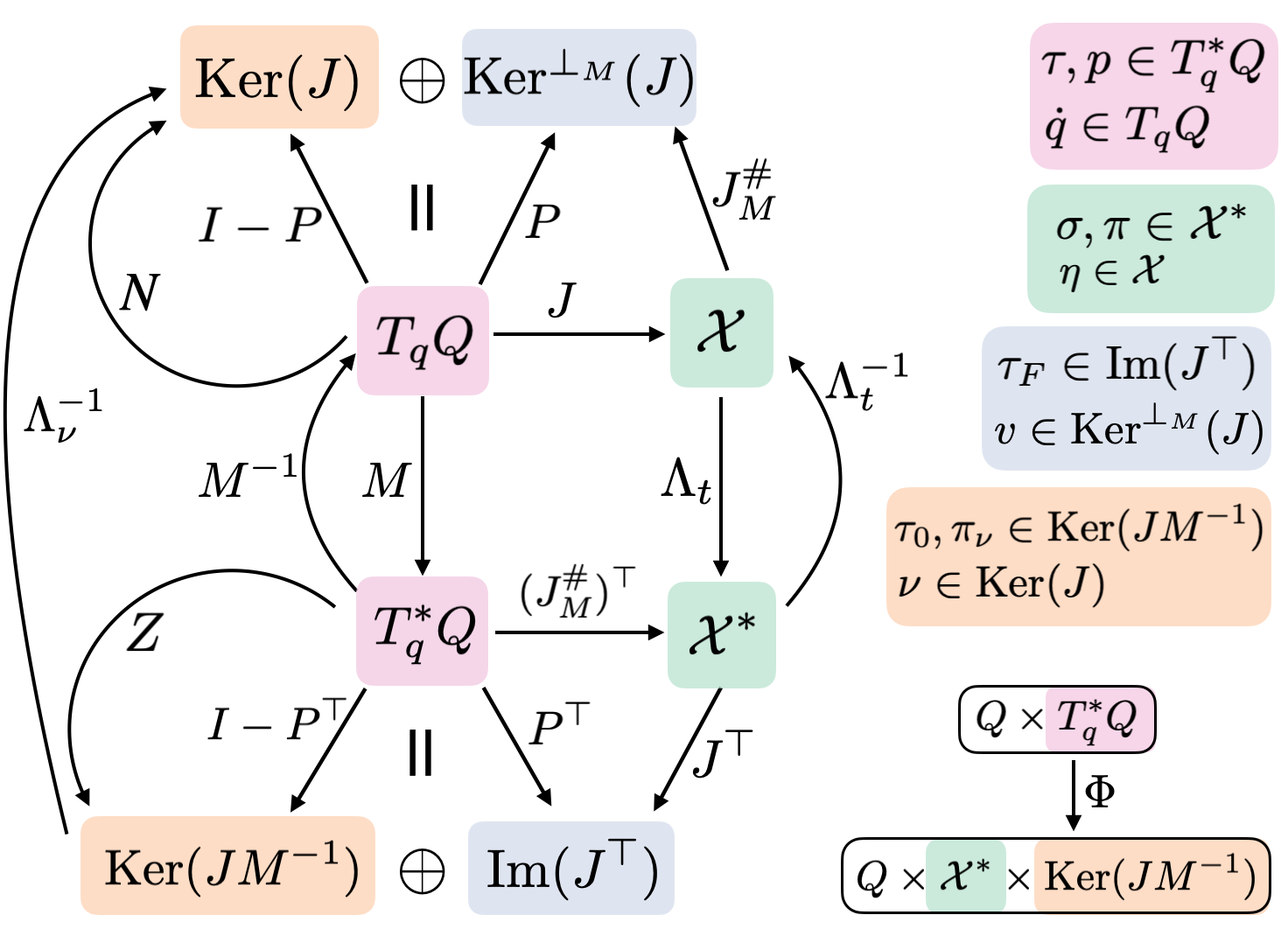}
    \caption{A diagram illustrating the maps introduced in this work. Spaces highlighted with the same color contain power-conjugated dual variables.}
    \label{fig:placeholder}
\end{figure}
\section{IDA-PBC design for task equilibrium assignment and impedance shaping}
\label{sec:idapbc}
Now we will use the presented pH model to design a control law for redundant manipulators aiming at i) stabilizing a desired configuration in the task space, and ii) exploiting redundancy to shape desired impedance behavior at the task port. 
We will use the \textit{Interconnection and Damping Assignment - Passivity-Based Control} (IDA-PBC) approach, developed in the pH framework and aiming to shape, through state-feedback, the closed-loop dynamics as a desired pH structure.
As task, we consider an instance of (\ref{eq:task}) where $J$ is the analytic Jacobian of a redundant manipulator, i.e., $\eta=\dot{x}$, where $x\in\mathbb{R}^3$ are the chosen task-space position cartesian coordinates in euclidean space. We want to stabilize a given task equilibrium $x^*$ and assign a desired impedance profile that the manipulator exhibits in its neighborhood.
We take as objective the transformation of the system (\ref{extpH}) into the target system:
\begin{equation}
\label{eq:idapbctarget}
    \begin{pmatrix}
        \dot q\\ \dot \pi \\ \dot \pi_\nu 
    \end{pmatrix}=\underbrace{\begin{pmatrix}
        0 & J_{M}^\# & Z^T \\ -J_{M}^{\#^T} & \bar G_{11} & \bar G_{12}\\
        - Z & -\bar{G}_{12}^{\top} & \bar G_{22}
    \end{pmatrix}}_{J_d}\begin{pmatrix}
        \partial_q{H_{\textrm{cl}}}\\  \partial_\pi{H_{\textrm{cl}}}\\  \partial_{\pi_\nu}{H_{\textrm{cl}}}
    \end{pmatrix}
\end{equation}
where $H_{\textrm{cl}}$ is the new Hamiltonian function. Equating (\ref{extpH}) to (\ref{eq:idapbctarget}) we obtain
\begin{equation}
\label{eq:fullmatching}
    J_{z} \partial_z \bar{H} + \begin{pmatrix}
        0 & \sigma^{\top} & \tau_0^{\top}
    \end{pmatrix}^{\top}= J_d \partial_z H_{\textrm{cl}}.
\end{equation}
Pre-multiplying the latter by $g^\perp = \begin{bmatrix} I_n & 0 & 0 \end{bmatrix}$
we obtain the so-called \textit{matching conditions}, a set of nonlinear PDEs on the non actuated coordinates whose feasibility determines the possibility to successfully implement the IDA-PBC procedure. In this case we get:
\begin{equation}\label{matchingeq}
    J_{M}^{\#} (\partial_\pi \bar{H}) + Z^T (\partial_{\pi_\nu} \bar{H}) = J_{M}^\# (\partial_\pi H_{\textrm{cl}}) +  Z^T (\partial_{\pi_\nu} H_{\textrm{cl}})
\end{equation}
The simplest way to solve the matching conditions is to preserve the structure of the kinetic energy, i.e. to choose
\begin{equation}
\begin{cases}
    \partial_\pi H_{\textrm{cl}} = \partial_\pi \bar{H} \\
    \partial_{\pi_\nu} H_{\textrm{cl}} = \partial_{\pi_\nu} \bar{H}.
\end{cases}
\end{equation}
The target Hamiltonian will thus have the form:
\begin{equation}\label{desiredH}
   H_{\textrm{cl}}(z) = \bar{K}_t(q,\pi) + \bar{K}_\nu(q,\pi_\nu) + \bar{V}(q). 
\end{equation}
where $\bar{V}(q)=V(q)+V_c(q).$
The potential $V_c(q)$ represents the energy–shaping term introduced by the controller. 
Its role is to assign the desired equilibrium configuration $q^*$ as a critical point of the shaped energy $H_{\textrm{cl}}$.  
Accordingly, $V_c$ must be chosen to satisfy the \textit{extremum and minimum assignment conditions}:
\begin{equation}
\label{eq:minimumassignment}
    \partial_q \bar{V}(q^*) = 0, 
    \qquad 
    \partial_q^2 \bar{V}(q^*) \succ 0,
\end{equation}
so that $(q^*,0,0)$ becomes a strict local minimum of the closed–loop Hamiltonian $H_{\textrm{cl}}$. Moreover, the equilibrium should also satisfy the desired equilibrium assignment at a task level $ x^* = f(q^*)$.
By substituting the desired Hamiltonian \eqref{desiredH} into the matching equation \eqref{eq:fullmatching}, 
the control inputs can be explicitly determined, resulting in the state feedback law:
\begin{equation}
\label{eq:IDAPBC}
    \begin{pmatrix}
        \sigma \\
        \tau_0
    \end{pmatrix}=-\bar{J}^{-\top} (\partial_q V_c(q))+\begin{pmatrix}
        \bar G-G
    \end{pmatrix} \begin{pmatrix}
        \eta \\ \nu
    \end{pmatrix}-D\begin{pmatrix}
        \eta \\ \nu
    \end{pmatrix}
\end{equation}
where we collected the open loop $G_{ij}$ and closed-loop $\bar{G}_{ij}$ terms in skew-symmetric matrices $G$ and $\bar{G}$. The last addend is an extra control input 
(with $D=D^{\top}=\textrm{blockdiag}(D_t,D_\nu)\geq 0$) to inject artificial damping to increase the convergence rate to the closed-loop equilibrium, as clear from the closed-loop power balance:
\[ \dot{H}_{\textrm{cl}}=-
\begin{pmatrix}
        \eta \\ \nu
    \end{pmatrix}^{\top} D\begin{pmatrix}
        \eta \\ \nu
    \end{pmatrix}\leq0\]
The resulting stabilizing control law (\ref{eq:IDAPBC}) therefore consists of three distinct components: \textbf{(i)} \textit{Energy Shaping:} The term $-\bar{J}^{-\top} (\partial_q V_c(q))$  shapes the $H_{\textrm{cl}}$ in a  way to have the minimum at $z^*=(q^*,0,0)$; \textbf{(ii)} \textit{Energy Routing:} The second term modifies the interconnection matrix in the momenta dynamics. This term, whose design completely relies on the presented model, redistributes energy between the task and null-space ports in a power-preserving way, i.e., not contributing to the total power balance $\dot{H}_{\textrm{cl}}$; \textbf{(iii)} \textit{Damping Injection:}
    The last term introduces the dissipative effect required to asymptotically stabilize the equilibrium.

At this point we perform an impedance shaping step, beyond the just described stabilization. First we fix the energy routing component by choosing $\bar{G}_{11}=0$ and $\bar{G}_{12}=0$ in (\ref{eq:IDAPBC}) to simplify the impedance dynamics without compromising the stabilization. Since this external interaction enters the torque space as $J^{\top} F_{ext}$, it is easy to show that the closed-loop task momentum dynamics results in:
\begin{equation} \label{eq_task_ida}
    \dot \pi= -(J_M^\#)^\top \partial_q H_{\textrm{cl}} - D_t \partial_\pi H_{\textrm{cl}}+ F_{ext}
\end{equation}
Using $\pi=\Lambda_t  \dot{x}$, its time differentiation $ \dot \pi = \dot \Lambda_t \dot x+ \Lambda_t \ddot{x}$, we can rewrite \eqref{eq_task_ida} as:
\begin{equation}
\label{eq:impedanceclear}
    F_{ext}=  \Lambda_t \ddot{x}+ (\dot{\Lambda}_t+D_t)  \dot{x} + (J_M^\#)^\top (\partial_q {\bar H} + \partial_q \bar{V})
\end{equation}
which represents the nonlinear closed-loop impedance dynamics at the task port of interest. The part of the impedance that can be actively shaped by the proposed control design is the stiffness term, through the added potential $V_c(q)$. 
A clearer interpretation of the impedance behavior is obtained by linearizing (\ref{eq:impedanceclear}) around the task equilibrium $x^*$, characterized by the equilibrium state $z^*=(q^*,0,0)$ with a (non unique) $q^* \in f^{-1}(x^*)$, which is a minimum of $\bar{V}(q)$ as (\ref{eq:minimumassignment}) holds.
Using (\ref{eq:minimumassignment}) and introducing the perturbation variables $\delta x=x-x^*\approx J(q^*)\delta q$ and its dynamically consistent inverse $\delta q=J^{\#}_M(q^{*})\delta x$, the first order approximation around $x^{*}$ of (\ref{eq:impedanceclear}) results in:
\begin{equation}
    \Lambda_t(q^*) \delta \ddot x + D_t  \delta \dot{x} +\tilde{K}_{t}(q^*)  \delta x = F_{ext},
\end{equation}
where $\tilde{K}_{t}(q^*)=(J_M^{\#})^\top(q^*)\partial_q^2\bar{V}(q^*)J_M^{\#}(q^*).$
In Fig. \ref{fig:all_column} we report simulation performed on a 7-DOF Franka Emika Panda robot, with $x^*=(0.43,0,0.8)$. The degree of redundancy is $n-m=7-3=4$. We display simulations for two different instances of $q^*_k \in f^{-1}(x^*)$ ($k \in \{0,1 \}, q^*_0=[ 0,  -0.3,  0,  -1.5, 0,   1.5,  0] \textrm{ and }
q^*_1=[ 0.17,  1.69, -2.81, -1.47,  2.25,  3.67, -1.46]$) and added potential $V_c(q)=-V(q)+\frac{1}{2}(q-q^*)^{\top}W_k(q-q^*)$ (again $k \in \{0,1 \}$) so that $\partial_q^2\bar{V}(q^*)=W_k$ ($W_0=7 I_7$ is the soft shaping, $W_1=30 I_7$ is the stiff shaping). The results show the response to an end-effector force $F_{ext}=(10, 0, 0)$ (active for $0.15 s$, vertically shaded in the plots) of the stable closed-loop system controlled with (\ref{eq:IDAPBC}) for all four combinations of $q^*_k$ and $W_k$. Here $D_t=9 I_3$ and $D_\nu =6 I_4$, $G=\bar G$. Closed-loop stability is verified in all cases: $H_{\textrm{cl}}$ is always non increasing when the system is unperturbed. We display the evolution of the error of the norm $||\delta x||$ and the energies. 
\begin{figure}[ht]
    \vspace{-3mm}
    \centering
    \begin{subfigure}{\columnwidth}
        \centering
        \includegraphics[width=\linewidth]{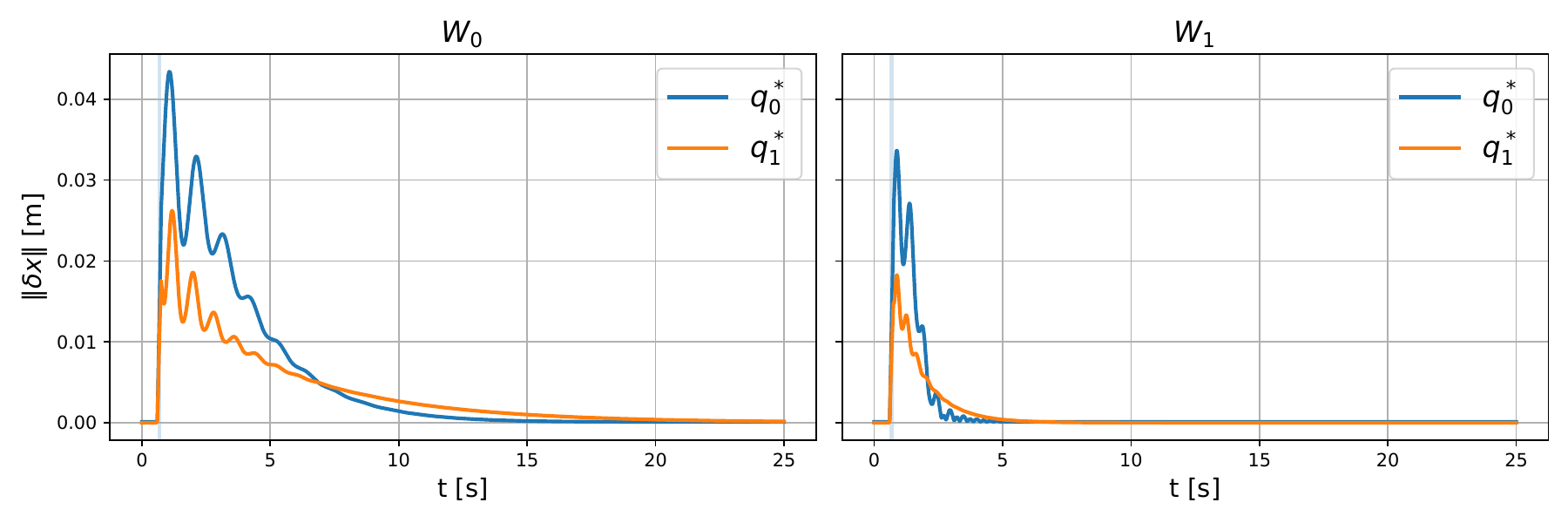}
        \label{fig:delta_u}
    \end{subfigure} 

 \vspace{-5mm}
    \begin{subfigure}{\columnwidth}
        \centering
        \includegraphics[width=\linewidth]{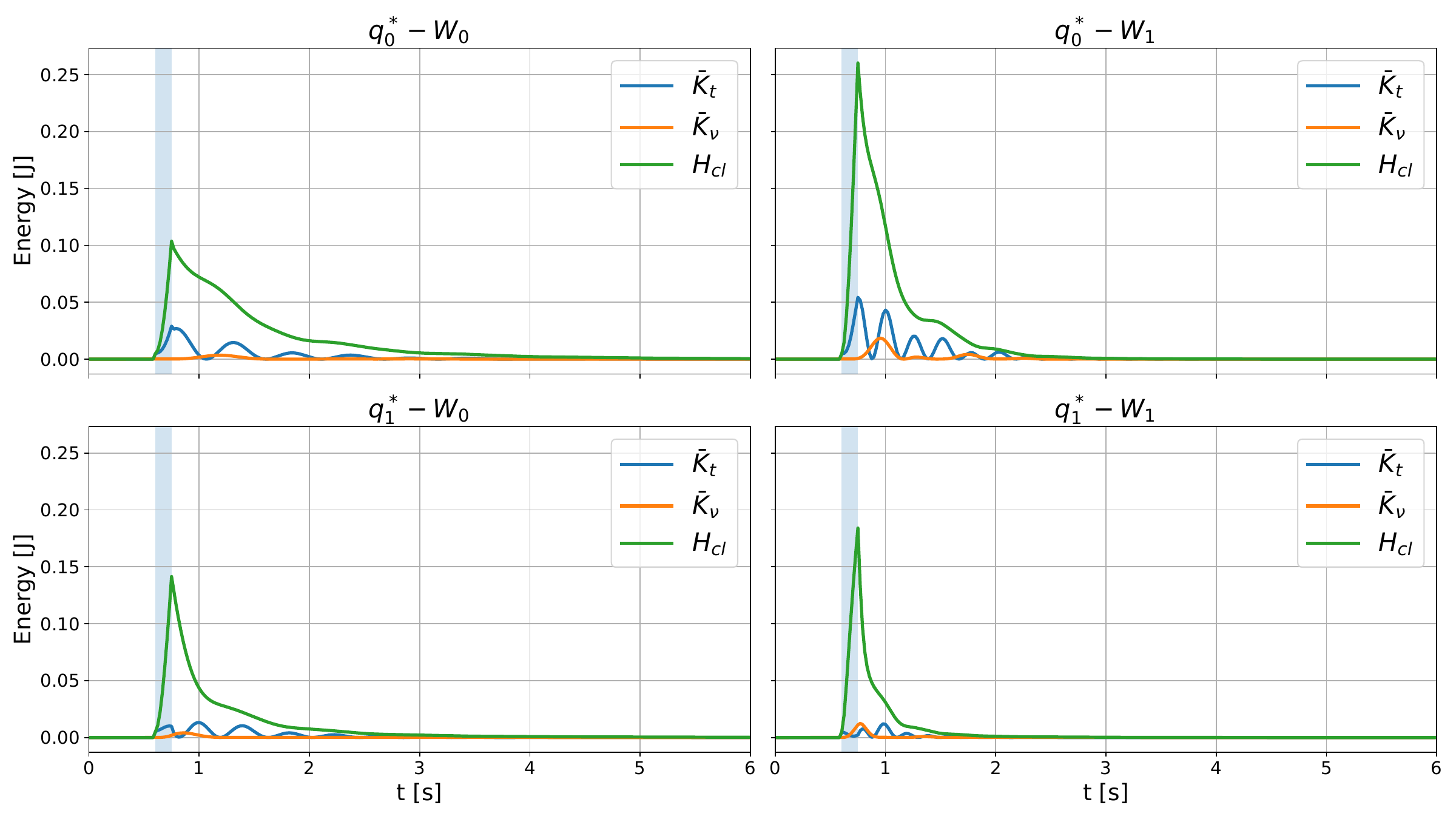}
        \label{fig:energies}
    \end{subfigure}
    \vspace{-10mm}

    \caption{Impedance response for different choices of $q^*$ and $V_c$; Task-space/Null-space/Total energy trajectories.}
    \label{fig:all_column}
\end{figure}
\vspace{-3.05mm}
\section{Conclusion}
\vspace{-2mm}
\label{sec:conc}
We presented a port-Hamiltonian formulation for redundant manipulators in the task space. The resulting model decomposes momentum variable and energy transfers into task-space and null-space components, enabling port-Hamiltonian modeling and control at the task port, while monitoring the energy transfers on the full system.
\bibliography{ifacconf}

\end{document}